\documentclass[conference, a4paper]{IEEEtran}

\IEEEoverridecommandlockouts
\usepackage{amsmath}
\usepackage{amsfonts}
\usepackage{amssymb}
\usepackage{graphicx}
\usepackage{epstopdf}%
\usepackage[nice]{nicefrac}
\usepackage{algorithmic}
\usepackage{algorithm}
\usepackage{array}
\usepackage{authblk}
\usepackage{tikz}
\usepackage[utf8]{inputenc}
\usepackage{pgfplots} 
\usepackage{pgfgantt}
\usepackage{pdflscape}
\usepackage{flushend}
\pgfplotsset{compat=newest} 
\pgfplotsset{plot coordinates/math parser=false}
\usepackage{grffile}
\usetikzlibrary{plotmarks}
\usepackage{amsthm}
\usepackage{amsmath}%
\usepackage{MnSymbol}%
\usepackage{wasysym}
\usepackage{subcaption}

\usepackage[mathscr]{euscript}
\usepackage{cite}
\usepackage{amsfonts}
\usepackage{pgf}
\usepackage{mathtools}
\usepackage{graphicx}
\usepackage{algorithmic}
\usepackage{algorithm}
\usepackage{array}
\usepackage{tikz}
\usepackage[utf8]{inputenc}
\usepackage{pgfplots} 
\usepackage{pgfgantt}
\usepackage{pdflscape}
\pgfplotsset{compat=newest} 
\pgfplotsset{plot coordinates/math parser=false}
\usepackage{grffile}
\usetikzlibrary{plotmarks}
\usepackage{tikzscale}
\usepackage{caption}
\usepackage{subcaption}
\usepackage{algorithmic}
\usepackage{algorithm}
\usepackage{array}
\usepackage{tikz}
\usepackage[utf8]{inputenc}
\usepackage{pgfplots} 
\usepackage{pgfgantt}
\usepackage{pdflscape}
\pgfplotsset{compat=newest} 
\pgfplotsset{plot coordinates/math parser=false}
\usepackage{grffile}
\usetikzlibrary{plotmarks}
\usepackage[T1]{fontenc}
\usepackage{graphicx}
\usepackage{caption}
\usepackage{overpic}
\usepackage{array}
\usepackage{color}
\usepackage{url}
\usepackage{tikz}
\usepackage[utf8]{inputenc}
\usepackage{pgfplots} 
\usepackage{cite}

\newtheorem{theorem}{Theorem}

\newtheorem{corollary}{Corollary}

\def\CN{\mathcal{C}\mathcal{N}} 

\begin{document}
	\title{Average Rate Analysis of RIS-aided Short Packet Communication in URLLC Systems}

	\author{Ramin Hashemi}
	\author{Samad Ali}
    \author{Nurul Huda Mahmood}
	\author{Matti Latva-aho.}
	\affil{Centre for Wireless Communications (CWC), University of Oulu, Oulu, Finland, \authorcr Emails: {\{ramin.hashemi, samad.ali, nurulhuda.mahmood,	matti.latva-aho\}@oulu.fi}}

	\maketitle
	
	\begin{abstract}
		In this paper, the average achievable rate of a re-configurable intelligent surface (RIS) aided factory automation is investigated in finite blocklength (FBL) regime. First, the composite channel containing the direct path plus the product of reflected paths through the RIS is characterized. Then, the distribution of the received signal-to-noise ratio (SNR) is matched to a Gamma random variable whose parameters depend on the total number of RIS elements as well as the channel pathloss. Next, by assuming FBL model, the achievable rate expression is identified and the corresponding average rate is elaborated based on the proposed SNR distribution. The phase error due to quantizing the phase shifts is considered in the simulation.  The numerical results show that Monte Carlo simulations conform to the matched Gamma distribution for the received SNR for large number of RIS elements. In addition, the system reliability indicated by the tightness of the SNR distribution increases when RIS is leveraged particularly when only the reflected channel exists. This highlights the  advantages of RIS-aided communications for ultra-reliable low-latency communications (URLLC) systems. The reduction of average achievable rate due to working in FBL regime with respect to Shannon capacity is also investigated as a function of total RIS elements.
	\end{abstract}
	
	
	\IEEEpeerreviewmaketitle

	\section{Introduction}
	\bstctlcite{IEEEexample:BSTcontrol}
	 In the sixth generation (6G) wireless networks \cite{MTCwhitePaper2020}, ultra-reliable low-latency communications (URLLC) \cite{Popovski2014,Popovski2019} will play an essential and inevitable role as the enabler of a wide range of applications such as industrial automation and e-health. The requirements are  stringent end-to-end delay and reliability of the networks. In URLLC systems, the high reliability means bit error rates (BER) as low as $10^{-9}$ and low latency is referred to the delay of less than 1 ms. URLLC messages usually carry control information, hence the packet lengths are generally ultra-short. As a result, the block length of the channel is short which necessitates a thorough analysis of achievable rate and decoding error probability as investigated in \cite{Polyanskiy2010,Yang2014}. The importance of URLLC is further highlighted in mission critical factory automation environments \cite{Luvisotto2017a}. 
	 \textcolor{black}{The three use-cases of URLLC including factory automation is introduced in \cite{Chen2018b}. In factory automation the design factors such as latency reduction techniques in the radio interface signaling should be well discussed to avoid link blockage. The requirements of Internet of things (IoT) applications including factory automation in terms of latency and reliability is discussed in \cite{Schulz2017}.} However, URLLC transmission demands are not met entirely as the main challenge to ensure high reliability is the random nature of the propagation channel due to multipath fading.  \\ 
    
    Recently, re-configurable intelligent surface (RIS) technology  \cite{DiRenzo2020,Wu2020b} is introduced that improves the spectral efficiency and coverage of wireless communication systems by influencing the propagation environment. The structure of an RIS is composed of a metasurface where a programmable controller configures and adjusts phase and/or amplitude response of the metasurface to modify the reflection of an incident electromagnetic wave towards the receiver. The aim of this operation is that the  received signals are added together coherently  so that the system performance improves in terms of enhancing e.g. the signal-to-noise ratio (SNR). Based on leveraging passive or active elements at each phase shifter, RISs are classified into passive and active devices, respectively.  Nevertheless, there are a number of challenges in RIS-aided systems such as acquiring the composite channel state information (CSI), and having finite resolution of the phase controller of the RIS which is due to limited quantization levels \cite{Wu2020b,Li2020}. Therefore, analyzing the statistical characteristics of RIS-empowered communications is of paramount importance. \\

    A number of studies have investigated the ergodic capacity or outage probability analysis of RIS-aided systems by identifying the characteristics of the channel response and received SNR \cite{Badiu2020,VanChien2020,Hou2020a,Qian2020}. In \cite{Badiu2020,VanChien2020} the SNR distribution is approximated as a Gamma random variable (RV), and the ergodic capacity is studied in an infinite blocklength channel \cite{VanChien2020} in which the composite channel contains the direct link plus the reflected signal from RIS with arbitrary phase shifts and as a result only statistical properties of the phase shifts was taken into account. 
    The best case and worst case channel responses are formulated as a Gamma RV with separate scale and shape parameters for each case in \cite{Hou2020a}. The authors in \cite{Qian2020} considered the optimal SNR derived in \cite{Zappone2020} and then, they proposed that the SNR distribution is composed of the product of three independent Gamma RVs and sum of two scaled non-central chi-square RVs based on the eigenvalues of the channel matrices of RIS-access point (AP) and RIS-user. The authors compare the proposed analytical distributions with the case that the SNR is only approximated with one gamma RV. The numerical results show that there is negligible difference in considering Gamma distribution for SNR compared with the most-exact analytical distributions. Furthermore, to evaluate the average rate, it is intractable to perform the expectations concerning SNR distribution when a complex expression is considered. Therefore, assuming the received SNR as a Gamma RV is tractable and accurate. \\
	

    On the other side, a number of studies investigate the performance analysis of URLLC systems in finite blocklength (FBL) regime such as  \cite{Li2018,Ren2020b,Tran2020}. The authors in \cite{Li2018} analyzed the ergodic achievable data rate at FBL channel model. Then, the optimal number of training symbols is studied based on the average data rate expression. In  \cite{Ren2020b} the authors proposed to employ massive multiple-input multiple-output (MIMO) systems to leverage in industrial IoT networks to reduce the latency. The lower bound achievable uplink ergodic rates of massive MIMO  system with FBL codes is analyzed by convexifying the rate formula which holds under specific conditions.
    In \cite{Tran2020} the downlink MIMO NOMA systems performance under Nakagami-m fading model, and average error probability is investigated in FBL regime. It should be noted that in \cite{Tran2020} the ergodic capacity analysis is not addressed and the error probability analysis is studied based on a well-known linear approximation for the Q-function. \\

    Even though the aforementioned work cover the topics of RIS and short-packet communication, to the best of our knowledge, there is no previous reports on the average achievable rate analysis of an RIS-aided transmission in a FBL scenario. This motivates us to shed some light on the average achievable analysis of RIS-aided factory automation wireless transmission under Rayleigh fading channel model which is assumed to be quasi-static during each transmission because of low velocity among APs and actuators (ACs). The presented analysis gives us design guidelines regarding the required number of RIS elements to ensure a target SNR or average achievable rate. Additionally, we show that leveraging RIS in industrial environments can greatly enhance the system reliability.\\

	In this paper, $\textbf{h} \sim \CN(\textbf{0}_{N\times 1},\textbf{C}_{N\times N})$ denotes circularly-symmetric (central) complex normal distribution vector with zero mean $\textbf{0}_{N\times 1}$ and covariance matrix $\textbf{C}$. The operators $\mathbb{E}[.]$ and $\mathbb{V}[.]$ denote the statistical expectation and variance, respectively. Also, $X\sim \Gamma(a,b)$ denotes Gamma random variable with scale and shape parameters $a$ and $b$, respectively.

	The structure of this paper is organized as follows. In Section II, the systems model and mathematical identification of the received SNR and its distribution is presented. In Section III the derivation of average rate is proposed. The numerical results are presented in Section IV. Finally, Section V concludes the paper.
	
	\section{System Model}
    Consider the downlink of an RIS-aided network consisting of a single antenna AP and AC where the RIS has $N = N_1 \times N_2$ elements. The channel response between the AP and AC have a direct component plus a reflected channel from the RIS. Let us denote the direct channel as $h_{\text{AC}}^{\text{AP}} \sim \CN(0,\eta^{\text{AP} \rightarrow \text{AC}})$ where $\eta^{\text{AP} \rightarrow \text{AC}}$ denotes the path loss attenuation due to large scale fading. $\textbf{h}_{\text{RIS}}^{\text{AP}} \in \mathbb{C}^{N \times 1}$ and $\textbf{h}_{\text{AC}}^{\text{RIS}} \in \mathbb{C}^{N \times 1}$ represent the vector channels from the AP to the RIS and from the RIS to the AC, respectively. The channel vector  $\textbf{h}_{\text{RIS}}^{\text{AP}}$ is distributed as $\CN(\textbf{0}_{N \times 1},\boldsymbol{\eta}^{\text{AP}\rightarrow\text{RIS}}_{N \times N})$  where $\boldsymbol{\eta}^{\text{AP}\rightarrow\text{RIS}}=\text{diag}(\eta^{\text{AP}\rightarrow\text{RIS}}_1,...,\eta^{\text{AP}\rightarrow\text{RIS}}_N)$ is a diagonal matrix including the path loss coefficients from the AP to the RIS elements. Similarly, the channel between the RIS and AC is distributed as $\textbf{h}_{\text{AC}}^{\text{RIS}} \sim \CN(\textbf{0}_{N \times 1},\boldsymbol{\eta}^{\text{RIS}\rightarrow\text{AC}}_{N \times N}) $ where  $\boldsymbol{\eta}^{\text{RIS}\rightarrow\text{AC}}=\text{diag}(\eta^{\text{RIS}\rightarrow\text{AC}}_1,...,\eta^{\text{RIS}\rightarrow\text{AC}}_N)$ denotes the covariance matrix in this case. In factory automation environments each actuator is almost in a fixed location. Therefore, the quasi-static channel fading model can be applied here. We assume that there is no interference for simplicity and leave the complex scenarios to future work.
	\begin{figure}
		\centering
		\includegraphics[scale=0.4]{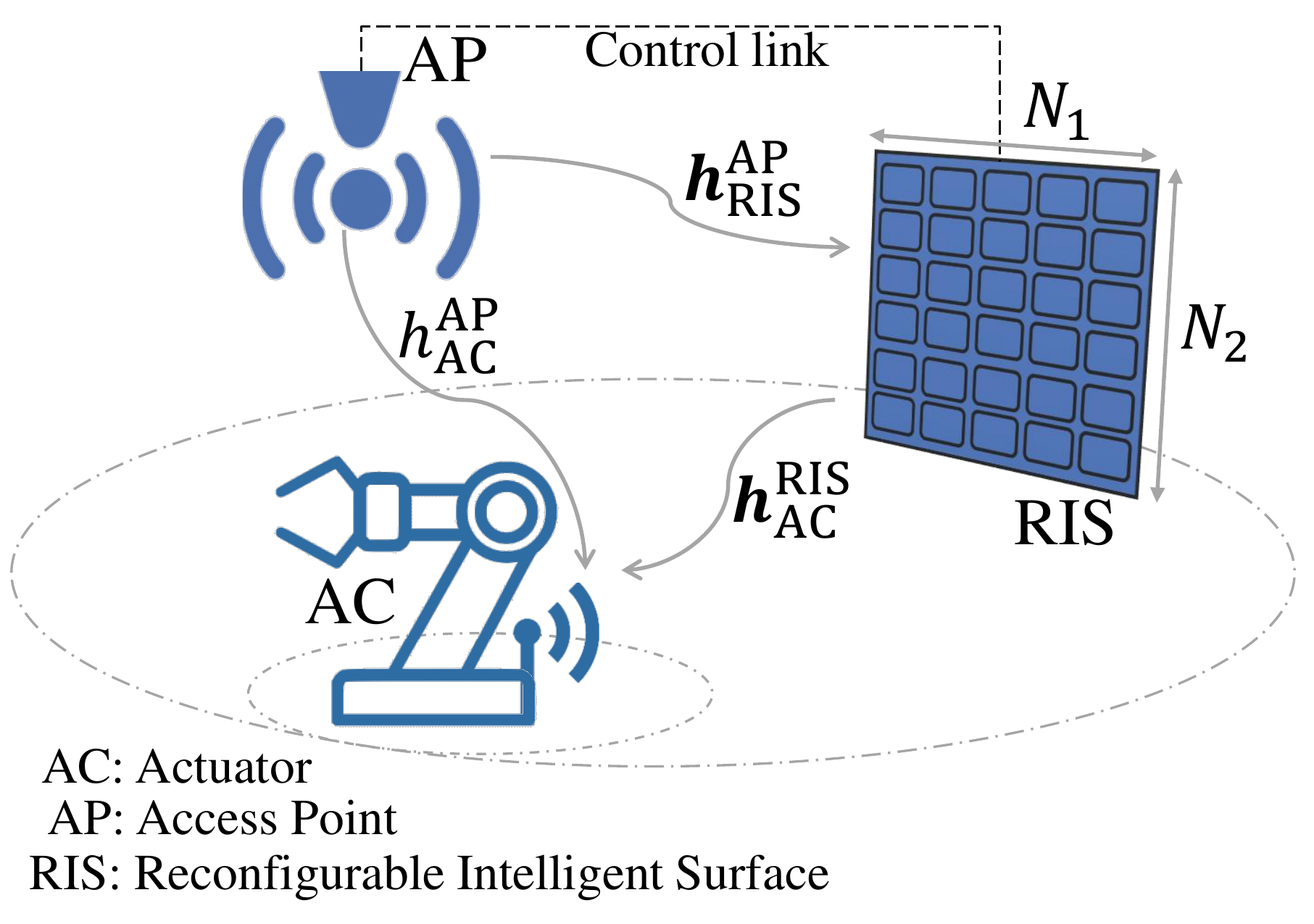}
		\caption{The system model.}
		\label{fig:2}
    \end{figure}
	
	The received signal at the AC is given by
	\begin{flalign}
	    y(t) = \left(h_{\text{AC}}^{\text{AP}} + {\textbf{h}_{\text{AC}}^{\text{RIS}}}^H\boldsymbol{\Theta}{\textbf{h}_{\text{RIS}}^{\text{AP}}}\right) s(t) +  n(t),
	\end{flalign}
	where $s(t)$ is the transmitted symbol from the AP with  $\mathbb{E}[|s(t)|^2] = p$ in which $p$ is the transmit power, and $n(t)$ is the additive white Gaussian noise with $\mathbb{E}[|n(t)|^2] = N_0 W$ where $N_0$, $W$ are the noise spectral density and the system bandwidth, respectively. The complex reconfiguration matrix $\boldsymbol{\Theta}_{N \times N}$ indicates the phase shift and the amplitude attenuation of RIS which is defined as 
	\begin{flalign}
	    \boldsymbol{\Theta}_{N \times N} = & \text{diag}(\beta_1 e^{j\theta_1},\beta_2 e^{j\theta_2},...,\beta_N e^{j\theta_N}), \nonumber \\ 
	    \beta_n \in & [0,1], \quad \forall n \in \mathcal{N} \\
	    \theta_n \in & [-\pi,\pi), \quad \forall n \in \mathcal{N} \nonumber 
	\end{flalign}
	where $\mathcal{N} = \{1,2,...,N\}$. Note that in our model we have assumed that the RIS elements have no coupling between them and there is no joint processing among elements. Hence, the phase shifts and amplitude control are done independently. In addition, we assume the phase shifts are performed without error and the received signals are coherently added at the receiver.

	Based on the received signal at AC and denoting the number of information bits $L$\footnote{It should be noted that $L$ that is the size of packets is assumed the same for actuator and access point.} that can be transmitted with target error probability $\varepsilon$ in $r$ channel uses ($r \geq 100$) the maximal achievable rate over a quasi-static additive white gaussian channel (AWGN) is given by \cite{Polyanskiy2010}
	\begin{flalign}
        R^{*}(\gamma,L,\varepsilon)=\frac{L}{r} =  \text{C}(\gamma) - Q^{-1}(\varepsilon)\sqrt{ \frac{\text{V}(\gamma)}{r}} + \mathcal{O}\left(\frac{\log_2(r)}{r}\right),
         \label{achievable_rate_urllc}
    \end{flalign}
    where $\text{C}(\gamma) = \log_2(1+\gamma)$ is the Shannon capacity formula under infinite blocklength assumption. The dispersion of the channel is defined as  $\text{V}(\gamma) = (\log_2(e))^2 \big( 1- \frac{1}{(1+\gamma)^2} \big)$. Note that $Q^{-1}(.)$ is the inverse of Q-function which is defined as $Q(x) = \frac{1}{\sqrt{2\pi}}\int_{x}^{\infty}e^{-\nu^2/2}d\nu$ and
    \begin{flalign}
    \gamma = \rho\left|h_{\text{AC}}^{\text{AP}} + {\textbf{h}_{\text{AC}}^{\text{RIS}}}^H\boldsymbol{\Theta}{\textbf{h}_{\text{RIS}}^{\text{AP}}}\right|^2,
    \end{flalign}
    where $\rho = \frac{p}{N_0W}$ denotes the instantaneous SNR. Note that the term $\mathcal{O}\big(\frac{\log_2(r)}{r}\big)$ in \eqref{achievable_rate_urllc} is neglected throughout this paper as it is approximately zero for $r \geq 100$ channel use. It is observed from \eqref{achievable_rate_urllc} when the blocklength approaches infinity the rate will be 
	\begin{flalign}
	    \lim_{r \rightarrow \infty} R^{*}(\gamma,L,\varepsilon) = \log_2\left(1+\rho\left|h^{\text{AP}}_{\text{AC}} + {\textbf{h}_{\text{AC}}^{\text{RIS}}}^H\boldsymbol{\Theta}{\textbf{h}_{\text{RIS}}^{\text{AP}}}\right|^2\right),
	\end{flalign}
	which is the conventional Shannon capacity formula.

	To determine the average achievable rate, we need to identify the distribution of $\gamma$. In the following, we present the related theorems and derive a closed-form and tractable approximation for the average rate.
	
	\begin{theorem}[SNR distribution]
	\label{SNR_distribution_theorem}
    Let $X = \left|h^{\text{AP}}_{\text{AC}} + {\textbf{h}_{\text{AC}}^{\text{RIS}}}^H\boldsymbol{\Theta}{\textbf{h}_{\text{RIS}}^{\text{AP}}}\right|^2$ and given $N>>1$, the distribution of $X$ is approximately matched to a Gamma random variable with the following parameters \cite{VanChien2020,Hou2020a}
    \begin{flalign}
       X \sim \Gamma(\alpha',\beta'), 
    \end{flalign}
    where $\alpha$ and $\beta$ are given in terms of first and second order moment of $X$ as 
    \begin{flalign}
    \label{alpha}
       \alpha' = \frac{(\mathbb{E}[X])^2}{\mathbb{E}[X^2]-(\mathbb{E}[X])^2},   \\  
       \beta' = \frac{\mathbb{E}[X]}{\mathbb{E}[X^2]-(\mathbb{E}[X])^2},
    \end{flalign}
    where $\mathbb{E}[X]$ and $\mathbb{E}[X^2]$ are given in \eqref{appndx_b_eq1} and \eqref{appndx_3}. For SNR distribution we have  $\gamma = \rho X$. Therefore,   $\mathbb{E}[\gamma] = \rho \mathbb{E}[X]$ and $\mathbb{E}[\gamma^2] = \rho^2\mathbb{E}[X^2]$
    which implies that $\gamma \sim \Gamma(\alpha,\beta)$ with the same $\alpha $ as in \eqref{alpha} and $\beta = \frac{\beta'}{\rho}$.
    \end{theorem}
    \begin{proof}
          Let us denote $\mathbb{E}[|h_{\text{AC}}^{\text{AP}}|^2]=\eta^{\text{AP} \rightarrow \text{AC}}=\varsigma $, $ \mathbb{E}[|[\textbf{h}_{\text{RIS}}^{\text{AP}}]_n|^2]=\eta^{\text{AP} \rightarrow \text{RIS}}=\varrho$ and $ \mathbb{E}[|[\textbf{h}_{\text{AC}}^{\text{RIS}}]_n|^2]= \eta^{\text{RIS} \rightarrow \text{AC}}=\vartheta$. When the phase adjustment is perfectly done at the RIS by employing the statistical analysis we can easily compute
    \begin{flalign}
        \mathbb{E}&[X] =  \varsigma + N \varrho\vartheta +  \frac{\pi^2N(N-1)}{16} \varrho\vartheta  + \frac{\pi N}{4}\sqrt{\varsigma\varrho\vartheta}, \label{appndx_b_eq1}\\ 
        \mathbb{E}&[X^2]  =  2\varsigma^2 + \varsigma\varrho\vartheta N\big(6  + \frac{3(N-1)\pi^2}{8} \big) \nonumber \\ 
        &  +  \frac{3N\pi^{1.5}}{4}\sqrt{\varsigma^3\varrho\vartheta} +
        \frac{\varrho^2\vartheta^2N}{256} \Big(\pi ^4 (N-3) (N-2) (N-1) \nonumber \\ 
        & + 48 \pi ^2 (2 N-1) (N-1)+768 N + 256 \Big) 
        \nonumber \\ 
        &
        + \sqrt{\varsigma\varrho^3\vartheta^3}\frac{N \pi^{1.5} }{32} \left( \pi ^2 (N-2) (N-1)+48 N-12 \right).
        \label{appndx_3}
    \end{flalign}
    therefore, the mean and variance of $\gamma = \rho X$ will be evaluated straightforwardly and accordingly the parameters of Gamma RV will be obtained which completes the proof.
    \end{proof}

    \section{Average rate analysis}
    In Theorem \ref{SNR_distribution_theorem}, we have modeled the SNR distribution, and the related Gamma distribution parameters $\alpha$ and $\beta$ are obtained. Next, to compute the average achievable rate, the instantaneous achievable rate should be averaged over the SNR distribution which we investigate in the next theorem. 
    \begin{theorem}
    \label{ergodic_rate_theo}
        The exact average achievable rate of the actuator in the RIS-aided FBL channel model given the distribution of SNR $\gamma \sim \Gamma(\alpha,\beta)$ is expressed as 
        \begin{flalign}
           \bar{R} = & r_1 -  \frac{Q^{-1}(\varepsilon)}{\sqrt{r}}r_2  \nonumber \\
           = &\frac{\beta^{\alpha}}{\Gamma(\alpha)\ln2} \sum_{k=1}^{\infty}\frac{1}{k}\Gamma(k+\alpha) \textbf{U}(k+\alpha,1+\alpha,\beta) \nonumber \\  & - \frac{Q^{-1}(\varepsilon)\beta^{\alpha}}{\sqrt{r}\ln2}\sum_{k=0}^{\infty}\binom{\frac{1}{2}}{k}(-1)^k \textbf{U}(\alpha,1-2k+\alpha,\beta),
        \end{flalign}
    where $r_1$ and $r_2$ are given in \eqref{r_1_summation_form} and \eqref{r2_expectation_binomial_series3}, respectively and $\textbf{U}(a,b,z)=\frac{1}{\Gamma(a)}\int_{0}^{\infty}(1+u)^{b-a-1}u^{a-1}e^{-zu}\,du$ denotes the confluent hypergeometric Kummer U function  \cite[Eq. (9.211)]{gradshteyn2014table}, and $\Gamma(\alpha) = \int_0^\infty y^{\alpha-1}e^{-y}\,dy$.
     \end{theorem}
     \begin{proof}
     The instantaneous rate is given by 
    	\begin{flalign}
            R^*(\gamma,L,\varepsilon) \approx \text{C}(\gamma) - Q^{-1}(\varepsilon)\sqrt{ \frac{\text{V}(\gamma)}{r}},
             \label{achievable_rate_urllc_v2}  
        \end{flalign}
        where $\mathcal{O}\big(\frac{\log_2(r)}{r}\big)$ is neglected as $r\geq 100$. To calculate the expected value of \eqref{achievable_rate_urllc_v2} in terms of the distribution of $\gamma$ we should compute the following 
        \begin{flalign}
            \bar{R} = \overset{r_1}{\overbrace{\mathbb{E}[\log_2(1+\gamma)]}} - \frac{Q^{-1}(\varepsilon)}{\sqrt{r}}\overset{r_2}{\overbrace{\mathbb{E}[\sqrt{\text{V}(\gamma)}] }} ,
            \label{expected_rate}
        \end{flalign}
        where we have used the linearity rule of the expectation. We investigate the two terms $r_1$ and $r_2$ involved in \eqref{expected_rate} separately. According to \eqref{expected_rate} $r_1$ is given by 
        \begin{flalign}
            r_1 = \mathbb{E}[\log_2(1+\gamma)] = \int_{0}^{\infty} \log_2(1+u)f_\gamma(u) \,du,
            \label{expectation_of_r1}
        \end{flalign}
        where $f_\gamma(u) = \frac{\beta^{\alpha}u^{\alpha-1}e^{-\beta u}}{\Gamma(\alpha)}$. To solve the integral, first consider the following series for natural logarithm \cite{gradshteyn2014table}
        \begin{flalign}
            \ln(1+x) = \sum_{k=1}^{\infty}\frac{1}{k}\left(\frac{x}{x+1}\right)^k, \quad x\geq 0
            \label{natural_logarithm_series}
        \end{flalign}
        then, by substituting in \eqref{expectation_of_r1} we will have
        \begin{flalign}
            & r_1 = \int_{0}^{\infty} \log_2(1+u)\frac{\beta^{\alpha}u^{\alpha-1}e^{-\beta u}}{\Gamma(\alpha)} \,du \nonumber \\ 
            = & \frac{1}{\ln2}\int_{0}^{\infty} \sum_{k=1}^{\infty}\frac{1}{k}\left(\frac{u}{u+1}\right)^k\frac{\beta^{\alpha}u^{\alpha-1}e^{-\beta u}}{\Gamma(\alpha)} \,du  \\
            \overset{a}{=} & \frac{\beta^{\alpha}}{\Gamma(\alpha)\ln2} \sum_{k=1}^{\infty}\frac{1}{k}\Gamma(k+\alpha) \int_{0}^{\infty}\frac{(1+u)^{-k}u^{\alpha+k-1}e^{-\beta u}}{\Gamma(k+\alpha)}\,du, \nonumber
        \end{flalign}
        where in $a$ the summation and the integral are exchanged and a constant factor of $\Gamma(k+\alpha)$ is multiplied in numerator and denominator. We observe that the integral inside the summation of the last step is defined as the confluent hypergeometric Kummer U function  $\textbf{U}(a,b,z)$ \cite[Eq. (9.211)]{gradshteyn2014table} therefore
        \begin{flalign}
            & r_1 = \frac{\beta^{\alpha}}{\Gamma(\alpha)\ln2} \sum_{k=1}^{\infty}\frac{1}{k}\Gamma(k+\alpha)\textbf{U}(k+\alpha,1+\alpha,\beta) , \label{r_1_summation_form}
        \end{flalign}
        where the summation can be truncated to a finite number in practical numerical situations. It should be noted that \eqref{r_1_summation_form} computes only the confluent hypergeometric function and, well-known gamma function at each iteration of the summation.

        Next, we evaluate the expectation associated with $r_2$ defined as
        \begin{flalign}
            r_2 = \mathbb{E}[\sqrt{\text{V}(\gamma)}] = \frac{1}{\ln2} \int_{0}^{\infty} \displaystyle \sqrt{1-\frac{1}{(1+u)^2}}f_\gamma(u) \,du,
            \label{expectation_of_r2}
        \end{flalign}
        to compute the above integral, we adopt the binomial expansion of channel dispersion which is given by
        \begin{flalign}
            \sqrt{\text{V}(\gamma)} = & \frac{1}{\ln2}\big( 1-\frac{1}{(1+\gamma)^2} \big)^{\frac{1}{2}} = \frac{1}{\ln2}\sum_{n=0}^{\infty}(-1)^n\binom{\frac{1}{2}}{n}(1+\gamma)^{-2n},\label{dispersion_binomial_series}
        \end{flalign}
        where $\binom{\frac{1}{2}}{n} = \frac{0.5(0.5-1)...(0.5-n+1)}{n!}$ for $n \neq 0$ and $\binom{\frac{1}{2}}{0} = 1$. Plugging \eqref{dispersion_binomial_series} in \eqref{expectation_of_r2} and substituting the definition of $f_\gamma(u)$ yields
        \begin{flalign}
            r_2 = & \frac{1}{\ln2}\int_{0}^{\infty} \sum_{n=0}^{\infty}(-1)^n\binom{\frac{1}{2}}{n}(1+u)^{-2n}f_\gamma(u) \,du, \nonumber \\ 
            = & \frac{\beta^{\alpha}}{\ln2}\sum_{n=0}^{\infty}\binom{\frac{1}{2}}{n}(-1)^n\textbf{U}(\alpha,1-2n+\alpha,\beta).
            \label{r2_expectation_binomial_series3}
        \end{flalign}
        consequently, if we substitute the mathematical expressions obtained in \eqref{r_1_summation_form} and  \eqref{r2_expectation_binomial_series3} in the average rate formula in \eqref{expected_rate}, the final result will be obtained which completes the proof.
     \end{proof}

       It is worth mentioning that the average rate given in Theorem \ref{ergodic_rate_theo} involves evaluating high-computational complexity functions as well as infinite summations which may not be beneficial in practical situations and resource allocation algorithms. Therefore, we propose a tractable lower bound approximation for the average rate in the following corollary. 
    
    \begin{corollary}
        A tractable and closed-form approximate lowerbound expression for the average rate $\bar{R}(L,\varepsilon)$ is given by
        \begin{flalign}
            \bar{R}_{\text{LB}}(L,\varepsilon) \approx & \Tilde{r}_1 - \frac{Q^{-1}(\varepsilon)}{\sqrt{r}}\Tilde{r}_2 =  \log_2\left(1+\frac{\alpha^2}{\beta(\alpha+1)}\right) \label{Corr_1} \\ \nonumber & -\frac{Q^{-1}(\varepsilon)}{2\ln2\sqrt{r}} \left(2-\beta +e^{\beta } \beta  (\alpha +\beta -1) \text{E}_{\alpha }(\beta )\right),
        \end{flalign}
        where $\text{E}_n(z)$ is the exponential integral function \cite[Eq. (8.211)]{gradshteyn2014table},  $\Tilde{r}_1$ and $\Tilde{r}_2$ are given in \eqref{Jensen_v2} and \eqref{expectation_of_r2_approximate_dispersion},  respectively.
    \end{corollary}
    \begin{proof}
    First, we study the term involving Shannon capacity, i.e. $r_1=\mathbb{E}[\log_2(1+\gamma)]$. Invoking  Jensen's inequality we have $r_1 = \mathbb{E}[\log_2(1+\gamma)] \geq \log_2(1+\nicefrac{1}{\mathbb{E}\Big[\nicefrac{1}{\gamma}\Big]})$. Then, we apply Taylor series expansion of $\frac{1}{\gamma}$ and take average from both sides where $\mathbb{E}[\frac{1}{\gamma}] \approx \frac{1}{\mathbb{E}[\gamma]} + \frac{\mathbb{V}[\gamma]}{\mathbb{E}[\gamma]^3} = \frac{\mathbb{E}[\gamma^2]}{\mathbb{E}[\gamma]^3}$  yields
    \begin{flalign}
        r_1 = \mathbb{E}[\log_2(1+\gamma)] \geq \tilde{r}_1 = \log_2\left(1+\frac{\alpha^2}{\beta(\alpha+1)}\right),
        \label{Jensen_v2}
    \end{flalign}
    where $\mathbb{E}[\gamma]^3$ and $\mathbb{E}[\gamma^2]$ can be evaluated straightforward from the SNR distribution $\gamma \sim \Gamma(\alpha,\beta)$ such that $\mathbb{E}[\gamma] = \frac{\alpha}{\beta}$ and $\mathbb{E}[\gamma^2] = \frac{\alpha(\alpha+1)}{\beta^2}$ where $\alpha$ and $\beta$ are investigated in Theorem \ref{SNR_distribution_theorem}.

    Next, we investigate $r_2$ in \eqref{expectation_of_r2}. To do so, we apply the truncated binomial approximation $(1+x)^\vartheta \approx 1+\vartheta x$ for $|x| < 1$, $|\vartheta x| < 1$ to approximate the channel dispersion as $\sqrt{\text{V}(u)}=\frac{1}{\ln2}\sqrt{1-\frac{1}{(1+u)^2}} \approx \frac{1}{\ln2}(1-\frac{1}{2(1+u)^2})$. 
    By substituting the approximate dispersion expression in \eqref{expectation_of_r2} we will have
    \begin{flalign}
        \tilde{r}_2 = & \mathbb{E}[\sqrt{\text{V}(\gamma)}] \approx \frac{1}{\ln2} \int_{0}^{\infty} \displaystyle (1-\frac{1}{2(1+u)^2})f_\gamma(u) \,du\nonumber \\  
         \overset{a}{=} &   \frac{1}{2\ln2} \left(2-\beta +e^{\beta } \beta  (\alpha +\beta -1) \text{E}_{\alpha }(\beta )\right).
            \label{expectation_of_r2_approximate_dispersion}
    \end{flalign}
    where $a$ is obtained through integrating by part and some manipulations. Finally, by replacing \eqref{Jensen_v2} and \eqref{expectation_of_r2_approximate_dispersion} in \eqref{expected_rate} the approximate result will be obtained. Note that since $r_1$ is lower bounded by $\tilde{r}_1$ and $r_2$ achieves its upper bounded by $\tilde{r}_2$ the achievable rate will attain its lower bound because of negative sign in $r_2$ in \eqref{expected_rate}. 
    \end{proof}

    \section{Numerical Results and Validation}
	In what follows, we evaluate the proposed derivations and mathematical expressions numerically. Table \ref{table2} shows the considered chosen values for the parameters of the network. 
	\begin{table}
		\caption{Simulation parameters.}
		\centering
		\begin{tabular}{ l  l }
			\hline
			Parameter & Default value \\ \hline
			Channel blocklength $r$ & 100 \\ The size of packets $L$ & $80$ bits   \\ 
			RIS location in 2D plane & ($d$,10) m $d \in [5,95]$ \\ AP location & (0,0) m \\
			AC position & (100,0) \\ AP transmit power $p$  & 200 mW  \\ 
			Receiver noise figure (NF)  & 3  dB \\ Target error probability $\varepsilon$  & $10^{-9}$ \\
			Noise power density $N_0$ & -174 dBm/Hz \\ 
			Number of realizations & $10^4$ \\ Bandwidth $W$ &  200 kHz \\ 
			Path loss model ($\mathscr{D}$: distance) & $\text{PL(dB)}=34.53+38\log_{10}(\mathscr{D})$ \\ \hline
		\end{tabular}
		\label{table2}
	\end{table}
	In Fig. \ref{fig:SNR_CDFa} and Fig. \ref{fig:SNR_CDFb} the cumulative distribution function of the received SNR is illustrated for two cases namely, with direct channel between the AP and the AC and the case where there is no direct channel. The latter can be, for example, due to blockage by large objects. Besides, the results are observed when perfect phase alignment is done at the RIS which is referred to as $\phi_n=0, \forall n \in \mathcal{N}$ as well as a uniformly distributed phase noise at the RIS resulting from quantization error. In other words, the RIS chooses each phase shift from set $\theta_n \in \Theta = \big\{-\pi,-\pi+\Delta,-\pi+2\Delta,...,-\pi+(2^b-1)\Delta\big\}, \smallskip \forall n \in \mathcal{N}$ where $\Delta = \frac{\pi}{2^{b-1}}$ and $b$ is the number of quantization bits.

	As we observe, Monte Carlo simulations conform to matched Gamma distribution to SNR in Fig. \ref{fig:SNR_CDFa}. Furthermore, the curves show that when RIS is employed, the slope of the CDF curves is much higher in comparison to having direct link and the RIS phase is not adjusted which is illustrated as $*$ and $**$ in Fig. \ref{fig:SNR_CDFb}. This interprets that when we have RIS and only the reflected channel exists we can expect almost a fixed SNR during transmission. This highlights the importance of RIS in increasing the reliability. On the other hand, when a direct link exists or the RIS phase is not adjusted as illustrated in Fig. \ref{fig:SNR_CDFb} we can expect a wide range of SNR values, which erodes the perception of reliability.

	Furthermore, the impact of quantizer noise is shown in Fig. \ref{fig:SNR_CDFa} and Fig. \ref{fig:SNR_CDFb}.  We observe that in the presence of the direct channel there is negligible difference between leveraging a 1-bit quantizer with a 2-bit in Fig. \ref{fig:SNR_CDFb}. This is because the direct channel has significantly higher power compared to the reflected path signal so that the effect of reflected channel is not dominant. On the other hand when there is no direct link the reduction of SNR due to quatization bits is to about 3.9 (0.9) dB for a 1-bit (2-bit) quantizer compared to the optimal case while a  3-bit quantizer is close to the perfect case with only a small SNR degradation. It can be inferred that utilizing a 2-bit quantizer leads to a good trade-off between RIS complexity, signaling overhead and performance.

	\begin{figure}[t]
		\centering
		\includegraphics[trim = 8.5cm 8.5cm 8.5cm 8.5cm,scale=0.59]{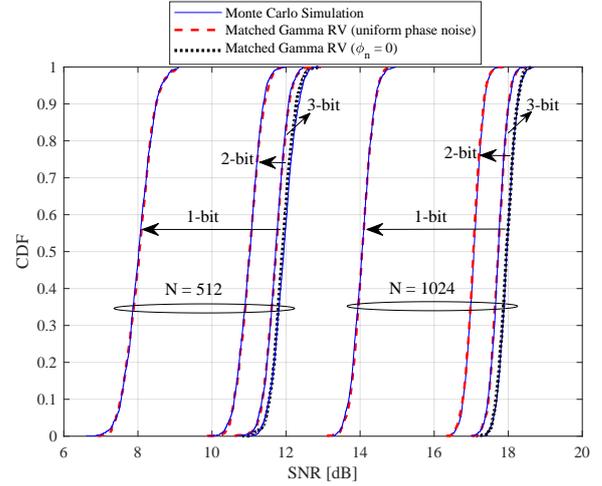}
		\caption{SNR CDFs without direct link.}
		\label{fig:SNR_CDFa}
	\end{figure}
	
	\begin{figure}[t]
		\centering    
		\includegraphics[trim = 8.5cm 8.5cm 8.5cm 8.5cm,scale=0.59]{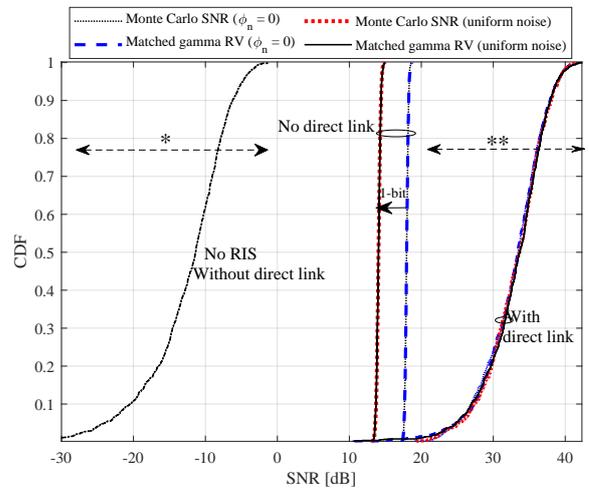}
		\caption{SNR CDFs with/without direct channel ($N = 1024$).}
		\label{fig:SNR_CDFb}
	\end{figure}
    
    By evaluating the average achievable rate as well as taking into account the channel dispersion, a 2-bit quantizer is compared with a 1-bit quantizer in Fig. \ref{fig:ErgodicRate} along with a perfect phase alignment scenario. The results also confirm that Monte Carlo simulations approximate to the derived analytical expressions for average rate in Theorem \ref{ergodic_rate_theo}.  As it is observed, when the number of bits assigned to each discrete phase at the RIS increments, the average rate curve is very close to the perfect phase alignment case. This shows that to achieve satisfactory accuracy, a few numbers of available bits will be sufficient  instead of high precision and high complexity quantizers. In addition, the performance gap between Shannon capacity and average rate in FBL regime is increased in terms of the number of RIS elements. However, the gap has reached to a fixed value in higher elements. This is because the  channel dispersion is saturated as the number of elements increase which results in higher SNR and $\lim_{\gamma \rightarrow \infty}\text{V}(\gamma) = (\log_2(e))^2$. 
    
   	\begin{figure}[t]
		\centering
		\includegraphics[trim = 8.5cm 8.5cm 8.5cm 8.5cm,scale=0.6]{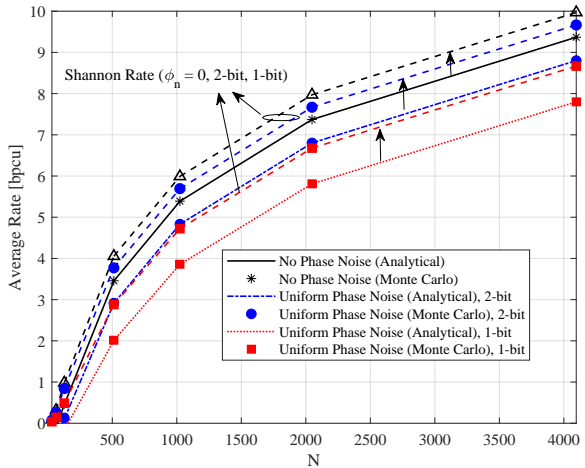}
		\caption{The average rate in terms of total RIS elements without presence of the direct channel.}
		\label{fig:ErgodicRate}
    \end{figure}

    Assuming the RIS is located at ($d$, 10) on the 2D plane  the average rate is illustrated versus $d \in [5,95]$ in Fig. \ref{fig:Ergodic_rate_vs_distance_N4096} with/without the presence of direct channel between AP and AC for $N=4096$. In contrast to the relaying schemes where locating in the middle of transmitter and the receiver is the best choice to maximize the performance, we can see that the average rate is maximized when the RIS is either close to the AP or the AC. Additionally assuming that we have direct link, we observe that the rate is varied from $9.9$ ($9.8$) bpcu to $11.25$ ($10.85$) bpcu when the RIS is re-located from the middle to the vicinity of AP or AC. This shows that the direct channel has a significant impact on the average rate as changing the RIS location has less impact on the average rate in the presence of direct link. 
   	\begin{figure}[t]
		\centering
		\includegraphics[trim = 8.5cm 8.5cm 8.5cm 8.5cm,scale=0.6]{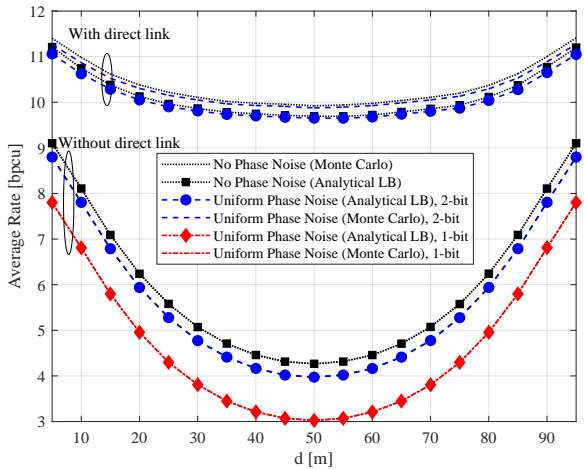}
		\caption{The impact of changing $d \in [5,95]$, where RIS is located at ($d$, 10) on the 2D plane on average rate in FBL regime.}
		\label{fig:Ergodic_rate_vs_distance_N4096}
    \end{figure}

    The average rate performance without direct channel is also  illustrated in Fig. \ref{fig:Ergodic_rate_vs_distance_N4096} where as shown in the curves there is a perfect match between the analytical expressions and the Monte Carlo simulations. The number of RIS elements is assumed the same as when direct channel exists ($N=4096$) to compare the results. Furthermore, the impact of phase error on the average rate is shown for 1-bit and 2-bit quantizers.  To have a similar analysis of rate variation we see that when $d=50$ is changed to $d=95$ or $d=5$ the average rate is increased from $4$ ($3$) bpcu to $9$ ($8$) bpcu for 2-bit (1-bit) quantizer. This shows a significant impact of changing the location of the RIS on the received SNR which affects the average rate compared to when direct channel exists particularly in factory automation scenarios.


    \section{Conclusion}
    In this paper, we have analysed the FBL regime performance of an RIS-aided communication link in a factory automation scenario. In particular, we have presented the analytical derivation of the achievable average rate and evaluated the impact of the quantization error. First, the received SNR is approximately matched to a Gamma RV whose parameters are derived in terms of total RIS elements and channels' path loss. Then, the average achievable rate is derived in terms of proposed SNR distribution. The numerical results have shown that the RIS can be effectively employed in factory automation environment for URLLC applications to enhance the reliability and improve the average achievable rate.
    
	\bibliographystyle{IEEEtran}
	\bibliography{IEEEabrv,refs} 

\begin{thebibliography}{10}
\providecommand{\url}[1]{#1}
\csname url@samestyle\endcsname
\providecommand{\newblock}{\relax}
\providecommand{\bibinfo}[2]{#2}
\providecommand{\BIBentrySTDinterwordspacing}{\spaceskip=0pt\relax}
\providecommand{\BIBentryALTinterwordstretchfactor}{4}
\providecommand{\BIBentryALTinterwordspacing}{\spaceskip=\fontdimen2\font plus
\BIBentryALTinterwordstretchfactor\fontdimen3\font minus
  \fontdimen4\font\relax}
\providecommand{\BIBforeignlanguage}[2]{{%
\expandafter\ifx\csname l@#1\endcsname\relax
\typeout{** WARNING: IEEEtran.bst: No hyphenation pattern has been}%
\typeout{** loaded for the language `#1'. Using the pattern for}%
\typeout{** the default language instead.}%
\else
\language=\csname l@#1\endcsname
\fi
#2}}
\providecommand{\BIBdecl}{\relax}
\BIBdecl

\bibitem{MTCwhitePaper2020}
\BIBentryALTinterwordspacing
N.~H. Mahmood \emph{et~al.}, \emph{White paper on critical and massive machine
  type communication towards {6G}}, ser. 6G Research Visions, nr. 11, N.~H.
  Mahmood \emph{et~al.}, Eds.\hskip 1em plus 0.5em minus 0.4em\relax Oulu,
  Finland: University of Oulu, Jun. 2020. [Online]. Available:
  \url{http://jultika.oulu.fi/files/isbn9789526226781.pdf}
\BIBentrySTDinterwordspacing

\bibitem{Popovski2014}
P.~Popovski, ``Ultra-reliable communication in {5G} wireless systems,'' in
  \emph{Proc. 1st Int. Conf. 5G Ubiquitous Connect.}\hskip 1em plus 0.5em minus
  0.4em\relax ICST, Feb. 2014, pp. 146--151.

\bibitem{Popovski2019}
P.~Popovski \emph{et~al.}, ``Wireless access in ultra-reliable low-latency
  communication ({URLLC}),'' \emph{IEEE Trans. Commun.}, vol.~67, no.~8, pp.
  5783--5801, May 2019.

\bibitem{Polyanskiy2010}
Y.~Polyanskiy, H.~V. Poor, and S.~Verd{\'{u}}, ``Channel coding rate in the
  finite blocklength regime,'' \emph{IEEE Trans. Inf. Theory}, vol.~56, no.~5,
  pp. 2307--2359, May 2010.

\bibitem{Yang2014}
W.~Yang, G.~Durisi, T.~Koch, and Y.~Polyanskiy, ``Quasi-static multiple-antenna
  fading channels at finite blocklength,'' \emph{IEEE Trans. Inf. Theory},
  vol.~60, no.~7, pp. 4232--4265, 2014.

\bibitem{Luvisotto2017a}
M.~Luvisotto, Z.~Pang, and D.~Dzung, ``Ultra high performance wireless control
  for critical applications: Challenges and directions,'' \emph{IEEE Trans.
  Ind. Informat.}, vol.~13, no.~3, pp. 1448--1459, Jun. 2017.

\bibitem{Chen2018b}
H.~Chen \emph{et~al.}, ``Ultra-reliable low latency cellular networks: Use
  cases, challenges and approaches,'' \emph{IEEE Commun. Mag.}, vol.~56,
  no.~12, pp. 119--125, Dec. 2018.

\bibitem{Schulz2017}
P.~Schulz \emph{et~al.}, ``Latency critical {IoT} applications in {5G}:
  Perspective on the design of radio interface and network architecture,''
  \emph{IEEE Commun. Mag.}, vol.~55, no.~2, pp. 70--78, Feb. 2017.

\bibitem{DiRenzo2020}
M.~{Di Renzo} \emph{et~al.}, ``Smart radio environments empowered by
  reconfigurable intelligent surfaces: How it works, state of research, and
  road ahead,'' \emph{IEEE J. Sel. Areas Commun.}, Apr. 2020.

\bibitem{Wu2020b}
Q.~Wu \emph{et~al.}, ``Intelligent reflecting surface aided wireless
  communications: A tutorial,'' \emph{IEEE Trans. Commun.}, pp. 1--1, 2021.

\bibitem{Li2020}
D.~Li, ``Ergodic capacity of intelligent reflecting surface-assisted
  communication systems with phase errors,'' \emph{IEEE Commun. Lett.},
  vol.~24, no.~8, pp. 1646--1650, Aug. 2020.

\bibitem{Badiu2020}
M.~A. Badiu and J.~P. Coon, ``Communication through a large reflecting surface
  with phase errors,'' \emph{IEEE Wireless Commun. Lett.}, vol.~9, no.~2, pp.
  184--188, Feb. 2020.

\bibitem{VanChien2020}
T.~{Van Chien}, L.~T. Tu, S.~Chatzinotas, and B.~Ottersten, ``Coverage
  probability and ergodic capacity of intelligent reflecting surface-enhanced
  communication systems,'' \emph{IEEE Commun. Lett.}, vol.~25, no.~1, pp.
  69--73, Jan. 2021.

\bibitem{Hou2020a}
T.~Hou \emph{et~al.}, ``Reconfigurable intelligent surface aided {NOMA}
  networks,'' \emph{IEEE J. Sel. Areas Commun.}, vol.~38, no.~11, pp.
  2575--2588, Nov. 2020.

\bibitem{Qian2020}
X.~{Qian} \emph{et~al.}, ``Beamforming through reconfigurable intelligent
  surfaces in single-user {MIMO} systems: {SNR} distribution and scaling laws
  in the presence of channel fading and phase noise,'' \emph{IEEE Wireless
  Commun. Lett.}, vol.~10, no.~1, pp. 77--81, Sep. 2021.

\bibitem{Zappone2020}
A.~Zappone \emph{et~al.}, ``Overhead-aware design of reconfigurable intelligent
  surfaces in smart radio environments,'' \emph{IEEE Trans. Wireless Commun.},
  vol.~20, no.~1, pp. 126--141, Jan. 2021.

\bibitem{Li2018}
C.~Li, S.~Yan, and N.~Yang, ``On channel reciprocity to activate uplink channel
  training for downlink wireless transmission in tactile internet
  applications,'' in \emph{2018 IEEE Int. Conf. Commun. Work. (ICC
  Work.)}.\hskip 1em plus 0.5em minus 0.4em\relax IEEE, May 2018, pp. 1--6.

\bibitem{Ren2020b}
H.~Ren \emph{et~al.}, ``{Joint pilot and payload power allocation for
  massive-MIMO-enabled URLLC IIoT networks},'' \emph{IEEE J. Sel. Areas
  Commun.}, vol.~38, no.~5, pp. 816--830, May 2020.

\bibitem{Tran2020}
\BIBentryALTinterwordspacing
D.-D. Tran \emph{et~al.}, ``Short-packet communications for {MIMO NOMA} systems
  over {Nakagami-m} fading: {BLER} and minimum blocklength analysis,'' Aug.
  2020. [Online]. Available: \url{http://arxiv.org/abs/2008.10390}
\BIBentrySTDinterwordspacing

\bibitem{gradshteyn2014table}
I.~S. Gradshteyn and I.~M. Ryzhik, \emph{Table of integrals, series, and
  products}.\hskip 1em plus 0.5em minus 0.4em\relax Academic press, 2007.

\end{thebibliography}
	
\end{document}